\documentclass[12pt,a4paper]{amsart}

\usepackage{amsmath,amsfonts,amssymb,color,extarrows}

\usepackage{undertilde}

\usepackage{tikz}

\usepackage[hmargin=2cm,vmargin=2cm]{geometry}

\usepackage{hyperref}
\usepackage{nameref,zref-xr}          

\usepackage{enumerate}     

\usepackage{cancel}

\newcommand{\mbP}{\mathbb P}
\newcommand{\mbZ}{\mathbb Z}
\newcommand{\mbC}{\mathbb C}
\newcommand{\cP}{\mathcal P}
\newcommand{\oM}{\overline{\mathcal M}}

\newcommand{\tu}{{\widetilde u}}
\newcommand{\og}{\overline g}
\newcommand{\oh}{\overline h}
\newcommand{\hLambda}{\widehat\Lambda}

\def\oM{{\overline{\mathcal{M}}}}
\def\CP{{{\mathbb C}{\mathbb P}}}
\renewcommand{\Im}{\mathrm{Im}}

\def\d{{\partial}}

\newcommand{\eps}{\varepsilon}

\newcommand{\cA}{\mathcal A}
\newcommand{\hcA}{\widehat{\mathcal A}}
\newcommand{\DR}{\mathrm{DR}}
\newcommand{\DZ}{\mathrm{DZ}}

\newcommand{\even}{\mathrm{even}}

\newcommand{\cF}{\mathcal F}

\newcommand{\Coef}{\mathrm{Coef}}

\newcommand{\Deg}{\mathrm{Deg}}

\newcommand{\tv}{\widetilde v}
\renewcommand{\top}{\mathrm{top}}
\newcommand{\red}{\mathrm{red}}

\newcommand{\gl}{\mathrm{gl}}

\newcommand{\mcF}{\mathcal{F}}

\newcommand{\cL}{\mathcal{L}}

\newcommand{\un}{{1\!\! 1}}

\newcommand{\cPO}{\mathcal{PO}}
\newcommand{\rt}{\mathrm{rt}}
\renewcommand{\t}{\mathrm{t}}

\newtheorem{theorem}{Theorem}[section]

\newtheorem{lemma}[theorem]{Lemma}

\newtheorem{conjecture}{Conjecture}

\theoremstyle{definition}

\newtheorem{definition}[theorem]{Definition}

\newtheorem{remark}[theorem]{Remark}

\newtheorem{convention}[theorem]{Convention}

\numberwithin{equation}{section}

\begin{document}

\title[The DR/DZ hierarchies at the approximation up to genus one]{The bihamiltonian structures of the DR/DZ hierarchies at the approximation up to genus one}

\author{Oscar Brauer}
\address{O. Brauer:\newline
School of Mathematics, University of Leeds, \newline
Leeds, LS2 9JT, United Kingdom}
\email{mmobg@leeds.ac.uk}

\author{Alexandr Buryak}
\address{A. Buryak:\newline 
Faculty of Mathematics, National Research University Higher School of Economics, \newline
6 Usacheva str., Moscow, 119048, Russian Federation;\smallskip\newline 
Center for Advanced Studies, Skolkovo Institute of Science and Technology, \newline
1 Nobel str., Moscow, 143026, Russian Federation}
\email{aburyak@hse.ru}

\begin{abstract}
In a recent paper, giving an arbitrary homogeneous cohomological field theory (CohFT), Rossi, Shadrin, and the first author proposed a simple formula for a bracket on the space local functionals that conjecturally gives a second Hamiltonian structure for the double ramification hierarchy associated to the CohFT. In this paper we prove this conjecture at the approximation up to genus~$1$ and relate this bracket to the second Poisson bracket of the Dubrovin--Zhang hierarchy by an explicit Miura transformation.
\end{abstract}

\date{\today}

\maketitle

\section{Introduction}

The appearance of integrable systems of PDEs in the intersection theory of the moduli spaces~$\oM_{g,n}$ of stable algebraic curves of genus $g$ with $n$ marked points was first manifested by the Kontsevich--Witten theorem \cite{Wit91,Kon92}, which states that the generating series of integrals over $\oM_{g,n}$ of monomials in psi classes (the first Chern classes of tautological line bundles) is controlled by a special solution of the Korteweg--de Vries hierarchy. Various versions of Witten's conjecture were proposed (the two most famous are in the Gromov--Witten theory of~$\mbC\mbP^1$~\cite{DZ04,OP06} and in the $r$-spin theory~\cite{Wit93,FSZ10}), when it was realized that integrable systems appear in a very general context, where the central role is played by the notion of a \emph{cohomological field theory} (CohFT). CohFTs are systems of cohomology classes on the moduli spaces~$\oM_{g,n}$ that are compatible with natural morphisms between the moduli spaces. They were introduced by Kontsevich and Manin in~\cite{KM94} to axiomatize the properties of Gromov--Witten classes of a given target variety. A \emph{correlator} of a CohFT is the integral over~$\oM_{g,n}$ of a monomial in the psi classes multiplied by a cohomology class forming the CohFT.

\medskip

In~\cite{DZ98} Dubrovin and Zhang constructed a Hamiltonian hierarchy controlling the correlators of an arbitrary CohFT at the approximation up to genus~$1$ and proved the polynomiality of the Hamiltonians and of the Poisson bracket. Moreover, in the case of a homogeneous CohFT they endowed the hierarchy with a polynomial bihamiltonian structure (also at the approximation up to genus~$1$). In the subsequent paper~\cite{DZ01} Dubrovin and Zhang presented a construction of a bihamiltonian hierarchy, called now the \emph{Dubrovin--Zhang (DZ) hierarchy} or the \emph{hierarchy of topological type}, controlling the correlators (in all genera) of an arbitrary semisimple homogeneous CohFT. However, the polynomiality of the Hamiltonians and of the two Poisson brackets was left as an open problem.

\medskip

In~\cite{BPS12b} the authors extended the construction of the DZ hierarchy to an arbitrary, not necessarily semisimple or homogeneous, CohFT and proved the polynomiality of the Hamiltonians and of the Poisson structure in the semisimple case (a simpler proof was obtained in~\cite{BPS12a}). In the case of a homogeneous CohFT the hierarchy is endowed with a second Hamiltonian structure whose polynomiality remains an important unproven feature of the DZ hierarchy.

\medskip

In~\cite{Bur15} a new construction of a Hamiltonian hierarchy associated to an arbitrary, not necessarily semisimple, CohFT was introduced. This construction is also based on the intersection theory on $\oM_{g,n}$, but it employs different tautological classes, notably the \emph{double ramification cycle} (an appropriate compactification of the locus of smooth curves whose marked points support a principal divisor), which explains why this hierarchy was called the \emph{double ramification (DR) hierarchy}. By the construction, the Hamiltonians of the DR hierarchy are polynomial, and the Poisson bracket is very simple, moreover, as opposed to the one for the DZ hierarchy, it does not essentially depend on the underlying CohFT. The two hierarchies coincide in the dispersionless (genus $0$) limit and, by a conjecture from~\cite{Bur15} called the \emph{DR/DZ equivalence conjecture}, they are related by a Miura transformation, which was completely identified in~\cite{BDGR18}. Although still unproved, the DR/DZ equivalence conjecture has accumulated a remarkable amount of evidence and verifications (see, e.g., \cite{BR16,BG16,BDGR18,BDGR20,BGR19,DR19}). In particular, the DR/DZ equivalence conjecture is proved at the approximation up to genus~$1$~\cite{BDGR18}.

\medskip

\begin{remark}\label{remark:remark about semisimplicity in genus 1}
Formally speaking, the semisimplicity assumption is present in the statement of Theorem~8.4 in~\cite{BDGR18} claiming that the DR/DZ equivalence conjecture is true at the approximation up genus~$1$. However, this assumption is never used in the proof. So the DR/DZ equivalence conjecture is true for an arbitrary CohFT at the approximation up genus~$1$.
\end{remark}

\medskip

In~\cite{BRS21}, giving an arbitrary homogeneous CohFT, the authors proposed a simple formula for a bracket on the space of local functionals and conjectured that it is Poisson and gives a second Hamiltonian structure for the DR hierarchy. These (conjecturally) Poisson brackets depend on the homogeneous CohFT under consideration in a remarkably explicit way. In this paper we prove the conjecture from~\cite{BRS21} at the approximation up to genus $1$. Moreover, we study a relation with the second Poisson bracket of the DZ hierarchy. The fact that the Hamiltonians and the first Poisson brackets of the DR and the DZ hierarchies are related by the Miura transformation described in~\cite{BDGR18} was proved in~\cite{BDGR18} at the approximation up to genus $1$. Here we check that this Miura transformation also relates the second Poisson brackets (at the approximation up to genus~$1$).

\medskip

\noindent{\bf Notations and conventions.} 
\begin{enumerate}[ \textbullet]
\item Throughout the text we use the Einstein summation convention for repeated upper and lower Greek indices.

\smallskip

\item When it doesn't lead to a confusion, we use the symbol $*$ to indicate any value, in the appropriate range, of a sub- or superscript.

\smallskip

\item For a topological space~$X$ let $H^*(X)$ denote the cohomology ring of~$X$ with coefficients in~$\mbC$.

\smallskip

\item For an integer $n\ge 1$ let $[n]:=\{1,2,\ldots,n\}$.
\end{enumerate}

\medskip

\noindent{\bf Acknowledgements.} The work of A.~B. has been funded within the framework of the HSE University Basic Research Program. O.~B. is supported by Becas CONACYT para estudios de Doctoradoen el extranjero awarded by the Mexican government, Ref: 2020-000000-01EXTF-00096.

\medskip

\section{Cohomological field theories}\label{subsection:CohFT}

Let $\oM_{g,n}$ be the Deligne--Mumford moduli space of stable algebraic curves of genus~$g$ with~$n$ marked points, $g\geq 0$, $n\geq 0$, $2g-2+n>0$. Note that $\oM_{0,3}=\mathrm{pt}$, and we will use the identification $H^*(\oM_{0,3})=\mbC$. Recall the following system of standard maps between these~spaces:
\begin{enumerate}[ \textbullet]
\item $\pi\colon\oM_{g,n+1}\to\oM_{g,n}$ is the map that forgets the last marked point. 

\smallskip

\item $\gl_{g_1,I_1;g_2,I_2}\colon\oM_{g_1,n_1+1}\times\oM_{g_2,n_2+1}\to \oM_{g_1+g_2,n_1+n_2}$, is the gluing map that identifies the last marked points of curves of genus $g_1$ and $g_2$ and turns them into a node. The sets~$I_1$ and~$I_2$ of cardinality $n_1$ and $n_2$, $I_1\sqcup I_2=[n_1+n_2]$, keep track of the relabelling of the remaining marked points. 

\smallskip

\item $\gl^{\mathrm{irr}}_{g,n+2}\colon\oM_{g,n+2}\to \oM_{g+1,n}$ is the gluing map that identifies the last two marked points and turns them into a node.
\end{enumerate}

\medskip

Let $V$ be a finite dimensional vector space of dimension $N$ with a distinguished vector $e\in V$, called the \emph{unit}, and a symmetric nondegenerate bilinear form $(\cdot,\cdot)$ on $V$, called the \emph{metric}. We fix a basis $e_1,\ldots,e_N$ in $V$ and let $\eta=(\eta_{\alpha\beta})$ denote the matrix of the metric in this basis, $\eta_{\alpha\beta}:= (e_\alpha,e_\beta)$, and let $A_\alpha$ denote the coordinates of $e$ in this basis, $e=A^\alpha e_\alpha$. As usual, $\eta^{\alpha\beta}$ denotes the entries of the inverse matrix, $(\eta^{\alpha\beta}):= (\eta_{\alpha\beta})^{-1}$.

\begin{definition}[\cite{KM94}]
A \emph{cohomological field theory} (CohFT) is a system of linear maps $c_{g,n}\colon V^{\otimes n} \to H^\even(\oM_{g,n})$, $2g-2+n>0$, such that the following axioms are satisfied:
\begin{enumerate}[ (1)]
\item The maps $c_{g,n}$ are equivariant with respect to the $S_n$-action permuting the $n$ copies of~$V$ in~$V^{\otimes n}$ and the $n$ marked points on curves from $\oM_{g,n}$, respectively;

\smallskip

\item $\pi^* c_{g,n}( \otimes_{i=1}^n e_{\alpha_i}) = c_{g,n+1}(\otimes_{i=1}^n  e_{\alpha_i}\otimes e)$ and $c_{0,3}(e_{\alpha_1}\otimes e_{\alpha_2} \otimes e)=\eta_{\alpha_1\alpha_2}$;

\smallskip

\item $\gl_{g_1,I_1;g_2,I_2}^* c_{g_1+g_2,n_1+n_2}( \otimes_{i=1}^{n_1+n_2} e_{\alpha_i}) = c_{g_1,n_1+1}(\otimes_{i\in I_1} e_{\alpha_i} \otimes e_\mu)\otimes c_{g_2,n_2+1}(\otimes_{i\in I_2} e_{\alpha_i}\otimes e_\nu)\eta^{\mu \nu}$;

\smallskip

\item $(\gl^\mathrm{irr}_{g,n+2})^* c_{g+1,n}(\otimes_{i=1}^n e_{\alpha_i}) = c_{g,n+2}(\otimes_{i=1}^n e_{\alpha_i}\otimes e_{\mu}\otimes e_\nu) \eta^{\mu \nu}$.
\end{enumerate}
\end{definition}

\medskip

Let us assume now that $V$ is a graded vector space and the basis $e_1,\ldots,e_N$ is homogeneous with $\deg e_\alpha = q_\alpha$, $\alpha=1,\dots,N$. Assume also that $\deg e = 0$. By $\Deg\colon H^*(\oM_{g,n})\to H^*(\oM_{g,n})$ we denote the operator that acts on~$H^i(\oM_{g,n})$ by multiplication by $\frac{i}{2}$.

\begin{definition} 
A CohFT $\{c_{g,n}\}$ is called \emph{homogeneous}, or \emph{conformal}, if there exist complex constants $r^\alpha$, $\alpha=1,\dots,N$, and $\delta$ such that
\begin{gather}\label{eq:definition of a homogeneous CohFT}
\Deg c_{g,n}(\otimes_{i=1}^ne_{\alpha_i})+\pi_*c_{g,n+1}(\otimes_{i=1}^ne_{\alpha_i}\otimes r^\gamma e_\gamma)=\left(\sum_{i=1}^n q_{\alpha_i}+\delta(g-1)\right)c_{g,n}(\otimes_{i=1}^ne_{\alpha_i}).
\end{gather}
The constant $\delta$ is called the \emph{conformal dimension} of CohFT.
\end{definition}

\medskip

For an arbitrary homogeneous CohFT let us introduce the formal power series
$$
F=F(t^1,\ldots,t^N):= \sum_{n\geq 3}\frac{1}{n!}\sum_{1\leq\alpha_1,\ldots,\alpha_n\leq N}\left(\int_{\oM_{0,n}}c_{0,n}(\otimes_{i=1}^n e_{\alpha_i})\right)\prod_{i=1}^n t^{\alpha_i}
$$
and define $C^\alpha_{\beta\gamma}:=\eta^{\alpha\nu}\frac{\d^3 F}{\d t^\nu\d t^\beta\d t^\gamma}$. The structure constants $C^\alpha_{\beta\gamma}$ define a formal family of commutative associative algebras with the unit $\frac{\d}{\d t^\un}:=A^\nu\frac{\d}{\d t^\nu}$, and moreover
$$
\left((1-q_\alpha)t^\alpha+r^\alpha\right)\frac{\d F}{\d t^\alpha}=(3-\delta)F+\frac{1}{2}A_{\alpha\beta}t^\alpha t^\beta,\quad\text{where}\quad A_{\alpha\beta}:= r^\mu c_{0,3}(e_\alpha\otimes e_\beta\otimes e_\mu),
$$
which means that the formal power series $F$ defines the structure of a homogeneous Dubrovin--Frobenius manifold~\cite{Dub96} on a formal neighbourhood of $0$ in $V$ with the Euler vector field given by $E=E^\alpha\frac{\d}{\d t^\alpha}:= \left((1-q_\alpha)t^\alpha+r^\alpha\right)\frac{\d}{\d t^\alpha}$. In particular, we have the following properties:
$$
C^\alpha_{\beta\gamma}C^\gamma_{\delta\theta}=C^\alpha_{\delta\gamma}C^\gamma_{\beta\theta},\qquad (\mu_\alpha+\mu_\beta)\eta_{\alpha\beta}=0.
$$

\medskip

\begin{convention}
We will systematically raise and lower indices in tensors using the metric $\eta$. For example, $C^{\alpha\beta}_\gamma:=\eta^{\alpha\nu}C^\beta_{\nu\gamma}$. 
\end{convention}

\medskip

\section{The Dubrovin--Zhang and the double ramification hierarchies}

\subsection{Differential polynomials, Poisson operators, and Hamiltonian hierarchies}

Let $u^1,\ldots,u^N$ be formal variables. Let us very briefly recall the main notions and notations in the formal theory of evolutionary PDEs with one spatial variable (and refer a reader, for example, to~\cite{BRS21} for details):
\begin{enumerate}[ \textbullet]

\item To the formal variables $u^\alpha$ we attach formal variables $u^\alpha_d$ with $d\ge 0$ and introduce the ring of \emph{differential polynomials} $\cA_u:=\mbC[[u^*]][u^*_{\ge 1}]$ (in~\cite{BRS21} it is denoted by~$\cA^0_u$). We identify $u^\alpha_0=u^\alpha$ and also denote $u^\alpha_x:=u^\alpha_1$, $u^{\alpha}_{xx}:=u^\alpha_2$, \ldots.

\smallskip

\item The space $\Lambda_u:=\left.\cA_u\right/(\mbC\oplus\Im\,\d_x)$ is called the space of \emph{local functionals} (in~\cite{BRS21} it is denoted by~$\Lambda^0_u$).

\smallskip

\item $\cA_{u;d}\subset\cA_u$ and $\Lambda_{u;d}\subset\Lambda_u$ are the homogeneous components of (differential) degree $d$, where $\deg u^\alpha_i:=i$.

\smallskip

\item The extended spaces of differential polynomials and local functionals are defined by $\hcA_u:= \cA_u[[\eps]]$ and $\hLambda_u:= \Lambda_u[[\eps]]$. Let $\hcA_{u;k}\subset\hcA_u$ and~$\hLambda_{u;k}\subset\hLambda_u$ be the subspaces of degree~$k$, where $\deg\eps:=-1$.  

\smallskip

\item We associate with $f\in\hcA_u$ the sequence of differential operators $L_\alpha^k(f):=\sum_{i\ge k}{i\choose k}\frac{\d f}{\d u^\alpha_i}\d_x^{i-k}$, $\alpha=1,\ldots,N$, $k\ge 0$. We denote $L_\alpha(f):=L^0_\alpha(f)$. 

\smallskip

\item Given an $N\times N$ matrix~$K=(K^{\mu\nu})$ of differential operators of the form $K^{\mu\nu} = \sum_{j\geq 0} K^{\mu\nu}_j \partial_x^j=\sum_{l,j\geq 0} \eps^l K^{[l],\mu\nu}_j \partial_x^j$, where $K^{[l],\mu\nu}_j\in\cA_{u;l-j+1}$, a bracket of degree~$1$ on the space $\hLambda_u$ is defined by $\{\overline{f},\overline{g}\}_{K}:=\int\left(\frac{\delta \overline{f}}{\delta u^\mu}K^{\mu \nu}\frac{\delta \overline{g}}{\delta u^\nu}\right)dx$.

\smallskip

\item An operator $K$ is called \emph{Poisson}, if the bracket $\{\cdot,\cdot\}_K$ is skewsymmetric and satisfies the Jacobi identity. The space of Poisson operators will be denoted by $\cPO_u$.	

\smallskip

\item Two Poisson operators $K_1$ and $K_2$ are said to be \emph{compatible} if the linear combination $K_2-\lambda K_1$ is a Poisson operator for any $\lambda\in\mbC$. 

\smallskip

\item A \emph{Miura transformation} is a change of variables $u^\alpha\mapsto \tu^\alpha(u^*_*,\eps)$ of the form $\tu^\alpha(u^*_*,\eps)=u^\alpha+\eps f^\alpha(u^*_*,\eps)$, where
$f^\alpha\in\hcA_{u;1}$. A Poisson operator $K$ rewritten in the new variables $\tu^\alpha$ will be denoted by $K_\tu$.

\smallskip

\item For a scalar operator $A=\sum_m A_m\d_x^m$, $A_m\in\cA_u$ (the sum is finite), let $A^\dagger:=\sum_m(-\d_x)^m\circ A_m$.

\smallskip

\item For a matrix operator $K=(K^{\alpha\beta})$, $K^{\alpha\beta}=\sum_m K^{\alpha\beta}_m\d_x^m$, $K^{\alpha\beta}_m\in\cA_u$ (the sum is finite), let $K^\dagger=(K^{\dagger;\alpha\beta})$, where $K^{\dagger;\alpha\beta}:=\sum_m(-\d_x)^m\circ K^{\beta\alpha}_m$.
\end{enumerate}

\medskip

\begin{definition}
A \emph{Hamiltonian hierarchy} of PDEs is a system of the form
\begin{gather*}
\frac{\partial u^\alpha}{\partial \tau_i} = K^{\alpha\mu} \frac{\delta\overline{h}_i}{\delta u^\mu}, \quad 1\le\alpha\le N,\quad i\ge 1,
\end{gather*}
where $\oh_i\in\hLambda_{u;0}$, $K=(K^{\mu\nu})$ is a Poisson operator, and $\{\oh_i,\oh_j\}_K=0$, $i,j\geq 1$. The local functionals~$\oh_i$ are called the \emph{Hamiltonians}.
\end{definition}

\medskip

\begin{definition} 
A Hamiltonian hierarchy of the form
\begin{gather}\label{eq:Hamiltonian system,2}
\frac{\d u^\alpha}{\d t^\beta_q} = K_1^{\alpha\mu}\frac{\delta\oh_{\beta,q}}{\delta u^\mu}, \quad 1\le\alpha,\beta\le N,\quad q\ge 0,
\end{gather}
equipped additionally with $N$ linearly independent Casimirs $\oh_{\alpha,-1}$, $1\le\alpha\le N$, of the Poisson bracket~$\{\cdot,\cdot\}_{K_1}$, is said to be \emph{bihamiltonian} if it is endowed with a Poisson operator~$K_2$ compatible with the operator~$K_1$ and such that
\begin{gather}\label{eq:bihamiltonian recursion}
\{\cdot,\oh_{\alpha,i-1}\}_{K_2}=\sum_{j=0}^i R^{j,\beta}_{i,\alpha}\{\cdot,\oh_{\beta,i-j}\}_{K_1},\quad 1\le\alpha\le N,\quad i\ge 0,
\end{gather}
where $R^j_i=(R^{j,\beta}_{i,\alpha})$, $0\le j\le i$, are constant $N\times N$ matrices. The relation~\eqref{eq:bihamiltonian recursion} is called a \emph{bihamiltonian recursion}. 
\end{definition}

\medskip

\subsection{The double ramification hierarchy}

Denote by $\psi_i\in H^2(\oM_{g,n})$ the first Chern class of the line bundle over~$\oM_{g,n}$ formed by the cotangent lines at the $i$-th marked point of stable curves. Denote by~$\mathbb E$ the rank~$g$ Hodge vector bundle over~$\oM_{g,n}$ whose fibers are the spaces of holomorphic one-forms on stable curves. Let $\lambda_j:= c_j(\mathbb E)\in H^{2j}(\oM_{g,n})$. 

\medskip

For any $a_1,\dots,a_n\in \mbZ$, $\sum_{i=1}^n a_i =0$, denote by $\DR_g(a_1,\ldots,a_n) \in H^{2g}(\oM_{g,n})$ the \emph{double ramification (DR) cycle}. We refer the reader, for example, to~\cite{BSSZ15} for the definition of the DR cycle on~$\oM_{g,n}$, which is based on the notion of a stable map to $\CP^1$ relative to~$0$ and~$\infty$. If not all the multiplicities $a_i$ are equal to zero, then one can think of the class $\DR_g(a_1,\ldots,a_n)$ as the Poincar\'e dual to a compactification in~$\oM_{g,n}$ of the locus of pointed smooth curves~$(C;p_1,\ldots,p_n)$ satisfying $\mathcal O_C\left(\sum_{i=1}^n a_ip_i\right)\cong\mathcal O_C$. 

\medskip

The crucial property of the DR cycle is that for any cohomology class $\theta\in H^*(\oM_{g,n})$ the integral $\int_{\oM_{g,n+1}}\lambda_g\DR_g\left(-\sum a_i,a_1,\ldots,a_n\right)\theta$ is a homogeneous polynomial in $a_1,\ldots,a_n$ of degree~$2g$ (see, e.g.,~\cite{Bur15}). Therefore, for a given CohFT $\left\{c_{g,n}\colon V^{\otimes n} \to H^{\even}(\oM_{g,n})\right\}$, define differential polynomials $g_{\alpha,d}\in\hcA_{u;0}$, $1\le\alpha\le N$, $d\ge 0$, as follows:
\begin{multline*}
g_{\alpha,d}:=\sum_{\substack{g,n\ge 0\\2g-1+n>0}}\frac{\eps^{2g}}{n!}\sum_{\substack{b_1,\ldots,b_n\ge 0\\b_1+\ldots+b_n=2g}}u^{\alpha_1}_{b_1}\ldots u^{\alpha_n}_{b_n}\times\\
\times\Coef_{a_1^{b_1}\ldots a_n^{b_n}}\int_{\oM_{g,n+1}}\DR_g\left(-\sum a_i,a_1,\ldots,a_n\right)\lambda_g\psi_1^d c_{g,n+1}(e_\alpha\otimes\otimes_{i=1}^n e_{\alpha_i}).
\end{multline*}
It is proved in~\cite{Bur15} that the local functionals~$\og_{\alpha,d}:=\int g_{\alpha,d}dx$ mutually commute with respect to the bracket~$\{\cdot,\cdot\}_{\eta^{-1}\d_x}$. 

\begin{definition}
The Hamiltonian hierarchy
$$
\frac{\d u^\alpha}{\d t^\beta_q}=\eta^{\alpha\mu}\d_x\frac{\delta\og_{\beta,q}}{\delta u^\mu},\quad 1\le\alpha,\beta\le N,\quad q\ge 0,
$$
is called the \emph{double ramification hierarchy}.
\end{definition}

\medskip

Let us equip the DR hierarchy with the following~$N$ linearly independent Casimirs of its Poisson bracket $\{\cdot,\cdot\}_{\eta^{-1}\d_x}$: $\og_{\alpha,-1}:=\int\eta_{\alpha\beta}u^\beta dx$, $1\le\alpha\le N$. Another important object related to the DR hierarchy is the local functional $\og\in\hLambda_{u;0}$ determined by the relation~\cite[Section 4.2.5]{Bur15} 
$$
\og_{\un,1}=(D-2)\og,\quad\text{where}\quad D:=\sum_{n\ge 0}(n+1)u^\alpha_n\frac{\d}{\d u^\alpha_n},
$$
and $\og_{\un,1}:=A^\alpha\og_{\alpha,1}$. Note that $\frac{\delta\og}{\delta u^\alpha}=g_{\alpha,0}$. Also, the local functional~$\og$ has the following explicit expression at the approximation up to~$\eps^2$ (\cite[Lemma~8.1]{BDGR18}):
\begin{gather}\label{eq:og up to genus 1}
\og=\int f dx-\frac{\eps^2}{48}\int c^\theta_{\theta\xi}c^\xi_{\alpha\beta}u^\alpha_xu^\beta_x dx+O(\eps^4),
\end{gather}
where $f:=\left.F\right|_{t^*=u^*}$ and $c^\alpha_{\beta\gamma}:=\left.C^\alpha_{\beta\gamma}\right|_{t^*=u^*}$.

\medskip

\begin{conjecture}[\cite{BRS21}]\label{conjecture:DR hierarchy} 
Consider a homogeneous CohFT and the associated DR hierarchy. 
\begin{enumerate}[ (1)]
\item The operator $K_2^\DR=\left(K_2^{\DR;\alpha\beta}\right)$ defined by
\begin{align}
K_2^{\DR;\alpha\beta}:=\eta^{\alpha\mu}\eta^{\beta\nu}&\left(\left(\frac{1}{2}-\mu_\beta\right)\d_x\circ L_\nu(g_{\mu,0})+\left(\frac{1}{2}-\mu_\alpha\right) L_\nu(g_{\mu,0})\circ\d_x\right.\label{eq:K2DR operator}\\
&\hspace{0.2cm}\left.+A_{\mu\nu}\d_x+\d_x\circ L_\nu^1(g_{\mu,0})\circ\d_x\right)\notag
\end{align}
is Poisson and is compatible with the operator $K^\DR_1:=\eta^{-1}\d_x$. Here $\mu_\alpha:= q_\alpha-\frac{\delta}{2}$.

\smallskip

\item The Poisson brackets $\{\cdot,\cdot\}_{K^\DR_2}$ and $\{\cdot,\cdot\}_{K^\DR_1}$ give a bihamiltonian structure for the DR hierarchy with the following bihamiltonian recursion:
\begin{gather*}
\left\{\cdot,\og_{\alpha,d}\right\}_{K^\DR_2}=\left(d+\frac{3}{2}+\mu_\alpha\right)\left\{\cdot,\og_{\alpha,d+1}\right\}_{K^\DR_1}+A^\beta_\alpha\left\{\cdot,\og_{\beta,d}\right\}_{K^\DR_1},\quad d\ge -1,
\end{gather*}
where $A^\alpha_\beta:=\eta^{\alpha\nu}A_{\nu\beta}$.
\end{enumerate}
\end{conjecture}

\medskip

\subsection{The Dubrovin--Zhang hierarchy}

Consider an arbitrary homogeneous CohFT $\{c_{g,n}\}$. Let $t^\alpha_a$, $1\le\alpha\le N$, $a\ge 0$, be formal variables, where we identify $t^\alpha_0=t^\alpha$. The \emph{potential} of our CohFT is defined by
$$
\mcF(t^*_*,\eps)=\sum_{g\ge 0}\eps^{2g}\mcF_g(t^*_*):=\hspace{-0.15cm}\sum_{\substack{g,n\ge 0\\2g-2+n>0}}\hspace{-0.05cm}\frac{\eps^{2g}}{n!}\hspace{-0.1cm}\sum_{\substack{1\leq\alpha_1,\ldots,\alpha_n\leq N\\d_1,\ldots,d_n\ge 0}}\hspace{-0.1cm}\left(\int_{\oM_{g,n}}\hspace{-0.3cm}c_{g,n}(\otimes_{i=1}^n e_{\alpha_i})\prod_{i=1}^n\psi_i^{d_i}\right)\prod_{i=1}^n t^{\alpha_i}_{d_i}\in\mbC[[t^*_*,\eps]],
$$
and introduce also the formal power series $w^{\top;\alpha}:=\eta^{\alpha\mu}\frac{\d^2\mcF}{\d t^\mu_0\d t^\un_0}$ and $w^{\top;\alpha}_n:=\frac{\d^n}{(\d t^\un_0)^n}w^{\top;\alpha}$, where $1\le\alpha\le N$ and $n\ge 0$.

\begin{conjecture}[\cite{DZ01}]\label{conjecture:DZ hierarchy}
Consider the ring $\hcA_{w}$ of differential polynomials in variables $w^1,\ldots,w^N$.
\begin{enumerate}[ (1)]

\item For any $1\le\alpha,\beta\le N$ and $a,b\ge 0$ there exists a differential polynomial $\Omega_{\alpha,a;\beta,b}\in\hcA_{w;0}$ such that 
\begin{gather}\label{eq:F and Omega}
\frac{\d^2\mcF}{\d t^\alpha_a\d t^\beta_b}=\left.\Omega_{\alpha,a;\beta,b}\right|_{w^\gamma_n=w^{\top;\gamma}_n}.
\end{gather}

\smallskip

\item There exists a Poisson operator $K_1^{\DZ}=\left(K_1^{\DZ;\alpha\beta}\right)$, for which the local functionals $\oh_{\alpha,-1}:=\int\eta_{\alpha\nu}w^\nu dx$ are Casimirs, such that 
\begin{gather}\label{eq:K1 and Omega}
\eta^{\alpha\mu}\d_x\Omega_{\mu,0;\beta,b}=K_1^{\DZ;\alpha\nu}\frac{\delta\oh_{\beta,b}}{\delta w^\nu},
\end{gather}
where $\oh_{\beta,b}:=\int\Omega_{\un,0;\beta,b+1}dx$, $1\le\alpha,\beta\le N$, $b\ge 0$. 

\smallskip

\item There exists a Poisson operator $K_2^{\DZ}=\left(K_2^{\DZ;\alpha\beta}\right)$ such that that following relations are satisfied:
\begin{gather}\label{eq:DZ recursion}
\left\{\cdot,\oh_{\alpha,d}\right\}_{K^\DZ_2}=\left(d+\frac{3}{2}+\mu_\alpha\right)\left\{\cdot,\oh_{\alpha,d+1}\right\}_{K^\DZ_1}+A^\beta_\alpha\left\{\cdot,\oh_{\beta,d}\right\}_{K^\DZ_1},\quad 1\le\alpha\le N,\quad d\ge -1.
\end{gather}
\end{enumerate}
\end{conjecture}

\medskip

If differential polynomials from Part 1 of the conjecture exist, then they are unique (see, e.g.,~\cite[Section~7.1]{BDGR18}). Moreover, if Poisson operators from Parts 2 and 3 exist, then they are also unique~(see, e.g.,~\cite[Section~6]{BPS12b}). Part 2 of the conjecture implies that the local functionals~$\oh_{\alpha,d}$ mutually commute with respect to the bracket~$\{\cdot,\cdot\}_{K_1^\DZ}$. The resulting bihamiltonian hierarchy (if the conjecture is true)
$$
\frac{\d w^\alpha}{\d t^\beta_q}=K_1^{\DZ;\alpha\mu}\frac{\delta\oh_{\beta,q}}{\delta u^\mu},\quad 1\le\alpha,\beta\le N,\quad q\ge 0,
$$
is called the \emph{Dubrovin--Zhang (DZ) hierarchy}. The $N$-tuple of formal power series $w^{\top;\alpha}$ is a solution of the hierarchy, where we identify the derivative $\d_x$ with $\frac{\d}{\d t^\un_0}$. This solution is called the \emph{topological solution}. 

\medskip

Conjecture~\ref{conjecture:DZ hierarchy} is proved at the approximation up to genus~$1$~\cite{DZ98}. In particular, 
\begin{align*}
&\Omega_{\alpha,a;\beta,b}=\left.\frac{\d^2\mcF_0}{\d t^\alpha_a\d t^\beta_b}\right|_{t^\gamma_c=\delta_{c,0}w^\gamma}+O(\eps^2),\\
&K_1^{\DZ;\alpha\beta}=\eta^{\alpha\beta}\d_x+O(\eps^2),\\
&K_2^{\DZ;\alpha\beta}=\left.\left(E^\gamma C_\gamma^{\alpha\beta}\right)\right|_{t^*=w^*}\d_x+\left(\frac{1}{2}-\mu_\beta\right)\left.\left(C^{\alpha\beta}_\gamma\right)\right|_{t^*=w^*}w^\gamma_x+O(\eps^2).
\end{align*}
Parts 1 and 2 of the conjecture are proved for an arbitrary semisimple, not necessarily homogeneous, CohFT \cite{BPS12b} (a considerably simplified proof of Part 2 is presented in~\cite{BPS12a}). 

\medskip

\subsection{The DR/DZ equivalence conjecture}

We again consider an arbitrary homogeneous CohFT. 

\medskip

The \emph{normal coordinates} of the DR hierarchy are defined by $\tu^\alpha(u^*_*,\eps):=\eta^{\alpha\mu}\frac{\delta\og_{\mu,0}}{\delta u^\un}$.

\medskip

In~\cite[Proposition~7.2]{BDGR18} the authors proved that there exists a unique differential polynomial $\cP\in\hcA_{w;-2}$ such that the power series $\cF^\red\in\mbC[[t^*_*,\eps]]$ defined by $\cF^\red:=\cF+\left.\cP\right|_{w^\gamma_n=w^{\top;\gamma}_n}$ satisfies the following vanishing property:
\begin{gather}\label{eq:property of Fred}
\Coef_{\eps^{2g}}\left.\frac{\d^n\cF^\red}{\d t^{\alpha_1}_{d_1}\ldots\d t^{\alpha_n}_{d_n}}\right|_{t^*_*=0}=0,\quad\text{if}\quad \sum_{i=1}^n d_i\le 2g-2.
\end{gather}
The differential polynomial $\cP$ has the following form: $\cP=-\eps^2 G(w^1,\ldots,w^N)+O(\eps^4)$, where $G(t^1,\ldots,t^N):=\left.\mcF_1\right|_{t^*_{\ge 1}=0}$. The power series $\cF^\red$ is called the {\it reduced potential} of our CohFT.

\medskip

Let us relate the variables $\tu^\alpha$ to the variables $w^\alpha$ by the following Miura transformation: $\tu^\alpha(w^*_*,\eps):=w^\alpha+\eta^{\alpha\nu}\d_x\{\cP,\oh_{\nu,0}\}_{K_1^\DZ}$.

\medskip

\begin{conjecture}[\cite{BDGR18} and \cite{BRS21}]\label{conjecture:DR/DZ equivalence}
Assuming that Conjectures~\ref{conjecture:DR hierarchy} and~\ref{conjecture:DZ hierarchy} are true, the DR and the DZ hierarchies, together with their bihamiltonian structures, coincide when we rewrite them in the coordinates~$\tu^\alpha$.
\end{conjecture}

\medskip

Our main result is the following theorem. 

\begin{theorem}\label{theorem:main}
Conjectures~\ref{conjecture:DR hierarchy} and~\ref{conjecture:DR/DZ equivalence} are true at the approximation up to genus $1$.
\end{theorem}

\medskip

The proof will be given in Section~\ref{section:proof of main theorem}. Together with Conjecture~\ref{conjecture:DZ hierarchy}, which was proved in~\cite{DZ98} at the approximation up to genus~$1$, the theorem gives a full understanding of the bihamiltonian structures of the DR and the DZ hierarchies and their relation at the approximation up to genus~$1$ for an arbitrary homogeneous CohFT. 

\medskip

\section{Extending the space of differential polynomials by tame rational functions}

Before proving Theorem~\ref{theorem:main} in the next section, let us present several technical lemmas.

\medskip

Conjecture~\ref{conjecture:DZ hierarchy} is not proved at the moment, but a weaker version is true if we extend the space of differential polynomials. Following~\cite[Section~7.3]{BDGR20} consider formal variables $v^1,\ldots,v^N$ and, for $d\in\mbZ$, denote by~$\cA^{\rt}_{v;d}$ the vector space spanned by expressions of the form
\begin{gather}\label{eq:rational function}
\sum_{i\ge m}\frac{P_i(v^*_*)}{(v^1_x)^i},
\end{gather}
where $m\in\mbZ$, $P_i\in\cA_{v;d+i}$ and $\frac{\d P_i}{\d v^1_x}=0$. Let $\cA^{\rt}_v:=\bigoplus_{d\in\mbZ}\cA^{\rt}_{v;d}$. Define also the extended space $\hcA^\rt_v:=\cA^\rt_v[[\eps]]$.

\medskip

A rational function~\eqref{eq:rational function} is called \emph{tame} if there exists a nonnegative integer~$C$ such that $\frac{\d P_i}{\d v^\alpha_k}=0$ for $k>C$. The subspace of tame elements in~$\cA^\rt_v$ will be denoted by $\cA^{\rt,\t}_v\subset\cA^\rt_v$. We also introduce the extended space $\hcA^{\rt,\t}_v:=\cA^{\rt,\t}_v[[\eps]]$. A \emph{rational Miura transformation} is a change of variables $v^\alpha\mapsto \tv^\alpha(v^*_*,\eps)$ of the form $\tv^\alpha(v^*_*,\eps)=v^\alpha+\eps f^\alpha(v^*_*,\eps)$, where $f^\alpha\in\hcA^{\rt,\t}_{v;1}$. 

\medskip

Introduce formal power series $v^{\top;\alpha}:=\eta^{\alpha\mu}\frac{\d^2\mcF_0}{\d t^\mu_0\d t^\un_0}$ and $v^{\top;\alpha}_n:=\frac{\d^n}{(\d t^\un_0)^n}v^{\top;\alpha}$. Note that the map $\hcA^{\rt,t}_v\to\mbC[[t^*_*,\eps]]$ given by
$$
\hcA^{\rt,t}_v\ni f\mapsto f|_{v^\gamma_c=v^{\top;\gamma}_c}\in\mbC[[t^*_*,\eps]]
$$
is injective \cite[Section~7.3]{BDGR20}. By~\cite[Proposition~7.6]{BDGR20} there exists a unique tame rational function $w^\alpha(v^*_*,\eps)\in\hcA^{\rt,t}_{v;0}$ such that $w^\alpha(v^{\top;*}_*,\eps)=w^{\top;\alpha}$. We also have $w^\alpha(v^*_*,\eps)-v^\alpha\in\Im\,\d_x$ (see, e.g., \cite[proof of Lemma~20]{BPS12b}). Note that the same proof as the proof of Proposition~7.6 in~\cite{BDGR20} also gives that there exists a unique tame rational function $\Omega_{\alpha,a;\beta,b}\in\hcA^{\rt,\t}_{w;0}$ such that equation~\eqref{eq:F and Omega} is true.

\medskip

Regarding Poisson operators, consider more general operators $K^{\mu\nu}=\sum_{j\ge 0}K^{\mu\nu}_j\d_x^j$ where $K^{\mu\nu}_j\in\hcA^{\rt,\t}_{w;-j+1}$. The space of such Poisson operators will be denoted by $\cPO_w^\rt$. Let $K_1^\DZ,K_2^\DZ\in\cPO_w^\rt$ be Poisson operators obtained from the operators 
$$
\eta^{\alpha\beta}\d_x \quad \text{and}\quad \left.\left(E^\gamma C_\gamma^{\alpha\beta}\right)\right|_{t^*=v^*}\d_x+\left(\frac{1}{2}-\mu_\beta\right)\left.\left(C^{\alpha\beta}_\gamma\right)\right|_{t^*=v^*}v^\gamma_x,
$$
respectively, by the rational Miura transformation $v^\alpha\mapsto w^\alpha(v^*_*,\eps)$. With this definition, the relations~\eqref{eq:K1 and Omega} and~\eqref{eq:DZ recursion} are true (see, e.g., \cite[Section~7.3]{BDGR20}). 

\medskip

Therefore, equivalently, Conjecture~\ref{conjecture:DZ hierarchy} says that $\Omega_{\alpha,a;\beta,b}\in\hcA_{w;0}$ and $K_1^\DZ,K_2^\DZ\in\cPO_w$.

\medskip

\begin{lemma}\label{lemma:constant terms of operators}
Consider a Poisson operator $K\in\cPO^\rt_v$ and a rational Miura transformation $v^\alpha\mapsto\tv^\alpha(v^*_*,\eps)$ such that $\tv^\alpha(v^*_*,\eps)-v^\alpha\in\Im\,\d_x$. Then we have $K^{\alpha\beta}_{\tv;0}=\sum_{m\ge 0}\frac{\d\tv^\alpha}{\d v^\rho_m}\d_x^m K_0^{\rho\beta}$.
\end{lemma}
\begin{proof}
We compute
\begin{align*} 
K^{\alpha\beta}_{\tv;0}=&\Coef_{\partial_{x}^{0}}K^{\alpha\beta}_{\tv}=\Coef_{\partial_{x}^{0}}\left(\sum_{m,n\ge 0}\frac{\partial \tv^{\alpha}}{\partial v_m^{\rho}} \partial_{x}^{m} \circ K^{\rho\theta} \circ (-\d_x)^n\circ\frac{\partial \tv^{\beta}}{\partial v_{n}^{\theta}}\right)=\\
=&\Coef_{\partial_{x}^{0}}\Bigg(\sum_{m\ge 0}\frac{\partial\tv^{\alpha}}{\partial v_m^{\rho}} \partial_{x}^{m} \circ K^{\rho\theta} \circ \underbrace{\sum_{n\ge 0}(-\d_x)^n\frac{\partial \tv^{\beta}}{\partial v_{n}^{\theta}}}_{=\frac{\delta \tv^\beta}{\delta v^\theta}=\delta^\beta_\theta}\Bigg)=\Coef_{\partial_{x}^{0}}\Bigg(\sum_{m\ge 0}\frac{\partial \tv^{\alpha}}{\partial v_m^{\rho}} \partial_{x}^{m} \circ K^{\rho\beta}\Bigg)=\\
=&\sum_{m\ge 0}\frac{\partial \tv^{\alpha}}{\partial v_m^{\rho}} \partial_{x}^{m}K^{\rho\beta}_0.
\end{align*}
\end{proof}

\medskip

\begin{lemma}\label{lemma:DZconstant}
We have $K_{2;0}^{\DZ;\alpha\beta}=\left(\frac{1}{2}-\mu_\beta\right)\eta^{\alpha\theta}\eta^{\beta\nu}\d_x\Omega_{\theta,0;\nu,0}$.
\end{lemma}
\begin{proof}
By Lemma~\ref{lemma:constant terms of operators} we have $K^{\DZ;\alpha\beta}_{2;0}=\left(\frac{1}{2}-\mu_\beta\right)\sum_{m\ge 0}\frac{\d w^\alpha}{\d v^\rho_m}\d_x^m\left(\left.\left(C^{\rho\beta}_\gamma\right)\right|_{t^*=v^*}v^\gamma_x\right)$. Note that $\left.\left(\left.\left(C^{\rho\beta}_\gamma\right)\right|_{t^*=v^*}v^\gamma_x\right)\right|_{v^*_*=v^{\top;*}_*}=\eta^{\beta\nu}\frac{\d v^{\top;\rho}}{\d t^\nu_0}$. Therefore,
$$
\left.\sum_{m\ge 0}\frac{\d w^\alpha}{\d v^\rho_m}\d_x^m\left(\left.\left(C^{\rho\beta}_\gamma\right)\right|_{t^*=v^*}v^\gamma_x\right)\right|_{v^*_*=v^{\top;*}_*}=\eta^{\beta\nu}\frac{\d w^{\top;\alpha}}{\d t^\nu_0}=\left.\eta^{\alpha\theta}\eta^{\beta\nu}\d_x\Omega_{\theta,0;\nu,0}\right|_{v^*_*=v^{\top;*}_*},
$$
which gives $\sum_{m\ge 0}\frac{\d w^\alpha}{\d v^\rho_m}\d_x^m\left(\left.\left(C^{\rho\beta}_\gamma\right)\right|_{t^*=v^*}v^\gamma_x\right)=\eta^{\alpha\theta}\eta^{\beta\nu}\d_x\Omega_{\theta,0;\nu,0}$, as required.
\end{proof}

\medskip

\begin{remark}
The lemma, in particular, implies that the constant term of the operator $K_2^{\DZ}$ is a differential polynomial if $\Omega_{\theta,0;\nu,0}$ is a differential polynomial, which is true in the semisimple case. This is also noticed in~\cite[Theorem~4.11]{HS21}. 
\end{remark}

\medskip

\begin{lemma}\label{lemma:DRconstant}
We have $K_{2;0}^{\DR;\alpha\beta}=\left(\frac{1}{2}-\mu_\beta\right)\eta^{\alpha\theta}\eta^{\beta\nu}\d_x\frac{\delta\og_{\nu,0}}{\delta u^\theta}$.
\end{lemma}
\begin{proof}
This directly follows from the definition~\eqref{eq:K2DR operator}.
\end{proof}

\medskip

\section{Proof of Theorem~\ref{theorem:main}}\label{section:proof of main theorem}

If we exclude the operators $K_2^\DZ$ and $K_2^\DR$ from consideration, then the fact that the DZ hierarchy and the DR hierarchy coincide in the coordinates $\tu^\alpha$, at the approximation up to genus~$1$, was already proved in~\cite[Theorem~8.4]{BDGR18} (see also Remark~\ref{remark:remark about semisimplicity in genus 1}). Thus, it is sufficient to prove that $K^{\DZ}_{2;\tu}=K^{\DR}_{2;\tu}+O(\eps^4)$. Since we know that $K^{\DZ;[0]}_{2;\tu}=K^{\DR;[0]}_{2;\tu}$~\cite[Proposition~2.1]{BRS21}, it remains to check that
\begin{gather}\label{eq:main equation}
K^{\DZ;[2]}_{2;\tu;l}=K^{\DR;[2]}_{2;\tu;l}\quad\text{for}\quad l=0,1,2,3.
\end{gather}
We split the proof in several steps. 

\medskip 

\noindent{\it\underline{Step 1}}. Let us check~\eqref{eq:main equation} for $l=2$ and $l=3$. We will do that by direct computation. 

\medskip

We will denote $\d_\alpha:=\frac{\d}{\d u^\alpha}$ and also use the notations
$$
c^{\alpha\beta}_{\gamma\delta}:=\d_\delta c^{\alpha\beta}_\gamma,\qquad c^{\alpha\beta}_{\gamma\delta\theta}:=\d_\theta c^{\alpha\beta}_{\gamma\delta},\qquad e^\gamma:=E^\gamma|_{t^*=u^*}=(1-q_\gamma)u^\gamma+r^\gamma,\qquad g^{\alpha\beta}=e^\gamma c^{\alpha\beta}_\gamma.
$$
In \cite[Thereom~2]{DZ98} the authors obtained the following formulas:
\begin{align*}
&K^{\DZ;[2],\alpha\beta}_{2;\tu;3}=\left.h^{\alpha\beta}\right|_{u^*=\tu^*},\\
&K^{\DZ;[2],\alpha\beta}_{2;\tu;2}=\left.\left(\frac{3}{2}\d_\gamma h^{\alpha\beta}+\frac{1}{24}\left(\frac{3}{2}-\mu_\beta\right)c^{\alpha\nu}_\gamma c^{\beta\mu}_{\nu\mu}-\frac{1}{24}\left(\frac{3}{2}-\mu_\alpha\right)c^{\beta\nu}_\gamma c^{\alpha\mu}_{\nu\mu}\right)\right|_{u^*=\tu^*}\tu^\gamma_x,
\end{align*}
where 
$$
h^{\alpha\beta}=\frac{1}{12}\left(\d_\nu\left(g^{\mu\nu}c^{\alpha\beta}_\mu\right)+\frac{1}{2}c^{\mu\nu}_\nu c^{\alpha\beta}_\mu\right).
$$

\medskip

On the other hand, since $\tu^\alpha=u^\alpha+\frac{\eps^2}{24}\d_x^2 c^{\alpha\mu}_\mu+O(\eps^4)$~\cite[proof of Theorem~8.4]{BDGR18}, we have
\begin{align*}
K_{2;\tu}^{\DR;\alpha\beta}=&L_{\nu}\left(u^{\alpha}+\frac{\varepsilon^{2}}{24} \partial_{x}^{2}c_{\lambda}^{\alpha\lambda}\right) \circ K_2^{\DR;\nu\rho} \circ L_{\rho}^\dagger\left(u^{\beta}+\frac{\varepsilon^{2}}{24} \partial_{x}^{2}c_{\theta}^{\beta \theta}\right)+O(\eps^4)=\\
=&K_2^{\DR;\alpha\beta}+\frac{\varepsilon^{2}}{24}\left(\d_x^2\circ L_\nu\left(c_{\lambda}^{\alpha \lambda}\right)\circ K_2^{\DR;\nu\beta}+K_2^{\DR;\alpha\rho}\circ L_\rho^\dagger\left(c_{\theta}^{\beta\theta}\right)\circ\d_x^2\right)+O(\eps^4)=\\
=&K_2^{\DR;\alpha\beta}+\varepsilon^{2}\underbrace{\frac{1}{24}\left(\d_x^2\circ c_{\nu\lambda}^{\alpha\lambda}\circ K_2^{\DR;[0],\nu\beta}+K_2^{\DR;[0],\alpha\rho}\circ c_{\rho\theta}^{\beta\theta}\circ\d_x^2\right)}_{=:\sum_{i=0}^3 R^{\alpha\beta}_i\d_x^i}+O(\eps^4),
\end{align*}
where $R^{\alpha\beta}_i\in\cA_{u;3-i}$. Considering the expansion $g_{\mu,0}=\sum_{g\ge 0}\eps^{2g}g^{[2g]}_{\mu,0}$, $g_{\mu,0}^{[2g]}\in\cA_{u;2g}$, from~\eqref{eq:K2DR operator} we compute
\begin{align*}
&K^{\DR;[2],\alpha\beta}_{2;3}=(3-\mu_\alpha-\mu_\beta)\eta^{\alpha\mu}\eta^{\beta\nu}\frac{\d g_{\mu,0}^{[2]}}{\d u^\nu_{xx}},\\
&K^{\DR;[2],\alpha\beta}_{2;2}=\eta^{\alpha\mu}\eta^{\beta\nu}\left[\left(2-\mu_\alpha-\mu_\beta\right)\frac{\d g_{\mu,0}^{[2]}}{\d u^\nu_{x}}+\left(\frac{5}{2}-\mu_\beta\right)\d_x\frac{\d g_{\mu,0}^{[2]}}{\d u^\nu_{xx}}\right].
\end{align*}
Using~\eqref{eq:og up to genus 1} we then compute
\begin{align*}
&\frac{\d g_{\mu,0}^{[2]}}{\d u^\nu_{xx}}=\frac{1}{24}c^\theta_{\theta\xi}c^\xi_{\mu\nu},\\
&\frac{\d g_{\mu,0}^{[2]}}{\d u^\nu_x}=\frac{1}{24}\left[\d_\nu\left(c^\theta_{\theta\xi}c^\xi_{\mu\gamma}\right)+\d_\gamma\left(c^\theta_{\theta\xi}c^\xi_{\mu\nu}\right)-\d_\mu\left(c^\theta_{\theta\xi}c^\xi_{\nu\gamma}\right)\right]u^\gamma_x,
\end{align*}
and finally get
\begin{align*}
&K^{\DR;[2],\alpha\beta}_{2;3}=\frac{3-\mu_{\alpha}-\mu_{\beta}}{24}c_\gamma^{\gamma\sigma}c^{\alpha\beta}_{\sigma},\\
&K^{\DR;[2],\alpha\beta}_{2;2}=\left[\frac{2-\mu_{\alpha}-\mu _{\beta}}{24}\left(c_{\theta \xi}^{\beta\theta} c_{\gamma}^{\alpha\xi}-c_{\theta \xi}^{\alpha\theta} c_{\gamma}^{\beta\xi}\right)+\frac{\frac{9}{2}-\mu_{\alpha}-2 \mu _{\beta}}{24} \partial_{\gamma}\left(c_{\theta}^{\theta \xi} c_{\xi}^{\alpha \beta}\right)\right]u_{x}^{\gamma}.
\end{align*}
Using that $K_2^{\DR;[0],\alpha\beta}=g^{\alpha\beta}\d_x+\left(\frac{1}{2}-\mu_\beta\right)c^{\alpha\beta}_\gamma u^\gamma_x$,  we also compute
\begin{align*}
&R_{3}^{\alpha \beta}=\frac{1}{24}\left(c^{	\alpha \lambda }_{\nu \lambda} g^{\nu \beta}+g^{\alpha \nu} c_{\nu \lambda}^{\beta \lambda}\right),\\
&R_{2}^{\alpha \beta}=\frac{1}{24}\left[ 2 \partial_{\gamma}\left( c^{\alpha \lambda}_{\nu \lambda} g^{\nu \beta}\right)+c_{\nu \lambda}^{\alpha \lambda} c_{\gamma}^{\nu \beta}\left(\frac{1}{2}-\mu_{\beta}\right)+g^{\alpha \nu} c_{\gamma \nu \lambda}^{\beta \lambda} +c_{\gamma}^{\alpha \nu}\left(\frac{1}{2}-\mu_{\nu}\right) c_{\nu\lambda}^{\beta \lambda}\right]u_x^\gamma.
\end{align*}	
So, in order to prove~\eqref{eq:main equation} for $l=2$ and $l=3$, we need to check the following two equations:
\begin{align}
&\frac{1}{12}\left(\partial_{\nu}\left(g^{\nu}{ }^{\mu} c_{\mu}^{\alpha \beta}\right)+\frac{1}{2} \underline{c_{\nu}^{\mu \nu} c_{\mu}^{\alpha \beta}}\right)=\frac{3-\mu_{\alpha}-\mu_{\beta}}{24}\underline{c_\gamma^{\gamma\sigma}c_\sigma^{\alpha\beta}}+\frac{1}{24}\left(c^{\alpha\lambda}_{\nu \lambda}g^{\nu \beta}+g^{\alpha \nu} c_{\nu\lambda}^{\beta\lambda}\right),\label{eq:l=3}\\
&\frac{1}{8} \partial_{\gamma}\bigg(\partial_{\nu}\left(g^{\mu \nu} c_{\mu}^{\alpha \beta}\right)+\frac{1}{2}\underbrace{{c}_{\nu}^{\mu \nu} c_{\mu}^{\alpha \beta}}_{*}\bigg)+\frac{1}{24}\left(\frac{3}{2}-\mu_{\beta}\right) \underbrace{c_{\gamma}^{\alpha \nu} c_{\nu \mu}^{\beta \mu}}_{**}-\frac{1}{24}\left(\frac{3}{2}-\mu_{\alpha}\right) \underbrace{c_{\gamma}^{\beta \nu} c_{\nu \mu}^{\alpha \mu}}_{***}=\label{eq:l=2}\\
=&\bigg[\frac{2-\mu_{\alpha}-\mu_{\beta}}{24}\Big(\underbrace{c_{\theta \xi}^{\beta\theta} c_{\gamma}^{\alpha\xi}}_{**}-\underbrace{c_{\theta \xi}^{\alpha\theta} c_{\gamma}^{\beta\xi}}_{***}\Big)+\frac{\frac{9}{2}-\mu_{\alpha}-2 \mu _{\beta}}{24}\partial_{\gamma}\big(\underbrace{c_{\theta }^{\theta \xi} c_{\xi}^{\alpha \beta}}_{*}\big)\bigg]\notag\\
&+\frac{1}{24}\bigg[2 \partial_{\gamma}\left(c_{\nu \lambda}^{\alpha \lambda} g^{\nu \beta}\right)+\left(\frac{1}{2}-\mu_{\beta}\right)\underbrace{c_{\nu\lambda}^{\alpha \lambda} c_{\gamma}^{\nu \beta}}_{***}+g^{\alpha \nu} c_{\gamma \nu \lambda}^{\beta \lambda}+\left(\frac{1}{2}-\mu_{\nu}\right)\underbrace{c_{\gamma}^{\alpha \nu}c_{\nu \lambda}^{\beta \lambda}}_{**}\bigg].\notag
\end{align}

\medskip

Let us prove~\eqref{eq:l=3}. Collecting the underlined terms together, using that $g^{\alpha\beta}=e^\gamma c_\gamma^{\alpha\beta}$, and multiplying both sides by~$12$, we see that~\eqref{eq:l=3} is equivalent to the following equation:
\begin{gather*}
(1-q_{\nu})\underline{c_{\nu}^{\nu \mu} c_{\mu}^{\alpha \beta}}+\boxed{e^{\gamma}\d_\nu\left(c_{\gamma}^{\nu \mu}c_{\mu}^{\alpha \beta}\right)}=\frac{2-\mu_{\alpha}-\mu_{\beta}}{2}\underline{c_{\gamma}^{\gamma\sigma} c_{\sigma}^{\alpha \beta}}+\frac{1}{2}e^{\gamma}\bigl(c_{\nu \lambda}^{\alpha \lambda} c_{\gamma}^{\nu\beta}+c_{\gamma}^{\alpha \nu} c_{\nu \lambda}^{\beta \lambda}\bigr).
\end{gather*}
Collecting the underlined terms together, transforming the boxed term as $e^{\gamma}\d_\nu\left(c_{\gamma}^{\nu \mu}c_{\mu}^{\alpha \beta}\right)=e^{\gamma}\d_\nu\left(c_{\gamma}^{\beta \mu}c_{\mu}^{\alpha\nu}\right)=e^{\gamma}\left(c_{\nu\gamma}^{\beta \mu}c_{\mu}^{\alpha\nu}+c_{\gamma}^{\beta \mu}c_{\nu\mu}^{\alpha\nu}\right)$, and moving all the terms to the left-hand side, we come to the expression
\begin{align}
&\frac{\mu_{\alpha} + \mu_{\beta}-2q_{\nu}}{2}c_{\nu}^{\nu \mu}c_{\mu}^{\alpha \beta}+e^{\gamma}\left(c_{\nu \gamma}^{\beta \mu} c_{\mu}^{\alpha \nu}+\underline{c_{\gamma}^{\beta \mu} c_{\nu\mu}^{\alpha\nu}}-\frac{1}{2} \underline{c_{\nu \lambda}^{\alpha \lambda} c_{\gamma}^{\nu \beta}}-\frac{1}{2}c_{\gamma}^{\alpha \nu} c_{\nu \lambda}^{\beta \lambda}\right)=\notag\\
=&\frac{\mu_{\alpha} + \mu_{\beta}-2q_{\nu}}{2}c_{\nu}^{\nu \mu}c_{\mu}^{\alpha \beta}+\boxed{e^{\gamma}c_{\nu \gamma}^{\beta \mu} c_{\mu}^{\alpha \nu}}+e^\gamma\left(\frac{1}{2}c_{\gamma}^{\beta \mu} c_{\nu\mu}^{\alpha\nu}-\frac{1}{2}c_{\gamma}^{\alpha \nu} c_{\nu \lambda}^{\beta \lambda}\right),\label{eq:tmp1}
\end{align}
whose vanishing we have to prove. From the theory of Dubrovin--Frobenius manifolds~\cite{Dub96} we know that $\cL_E C^\alpha_{\beta\gamma}=C^\alpha_{\beta\gamma}$, where $\cL_E$ denotes the Lie derivative, which implies that
\begin{equation*}
e^{\lambda} c_{\lambda \gamma}^{\alpha \beta}=\left(\delta-q_{\alpha}-q_{\beta}+q_{\gamma}\right) c_{\gamma}^{\alpha \beta}.
\end{equation*}
Applying this to the boxed term, we see that the expression~\eqref{eq:tmp1} is equal to
\begin{align}
&\frac{q_{\alpha} + q_{\beta}-2q_{\nu}-\delta}{2}\underline{c_{\nu}^{\nu \mu}c_{\mu}^{\alpha \beta}}+(\delta-q_\beta-q_\mu+q_\nu)\underline{c_{\nu }^{\beta \mu} c_{\mu}^{\alpha \nu}}+e^\gamma\left(\frac{1}{2}c_{\gamma}^{\beta \mu} c_{\nu\mu}^{\alpha\nu}-\frac{1}{2}c_{\gamma}^{\alpha \nu} c_{\nu \lambda}^{\beta \lambda}\right)=\notag\\
=&\frac{q_{\alpha}-q _{\beta} +\delta-2 q_{\mu}}{2}c_{\nu}^{\beta\mu}c_{\mu}^{\alpha \nu}+\frac{1}{2}e^\gamma\left(\underline{c_{\gamma}^{\beta \mu} c_{\nu\mu}^{\alpha\nu}-c_{\gamma}^{\alpha \nu} c_{\nu \lambda}^{\beta \lambda}}\right).\label{eq:tmp2}
\end{align}
We transform the underlined terms as follows:
$$
c_{\gamma}^{\beta \mu} c_{\nu\mu}^{\alpha\nu}-c_{\gamma}^{\alpha \nu} c_{\nu \lambda}^{\beta \lambda}=\left(\underline{\d_\nu(c_{\gamma}^{\beta \mu} c_{\mu}^{\alpha\nu})}-c_{\nu\gamma}^{\beta \mu} c_{\mu}^{\alpha\nu}\right)-\left(\underline{\d_\lambda(c_{\gamma}^{\alpha \nu} c_{\nu}^{\beta\lambda})}-c_{\lambda\gamma}^{\alpha \nu}c_{\nu}^{\beta \lambda}\right)=c_{\mu\gamma}^{\alpha \nu}c_{\nu}^{\beta \mu}-c_{\nu\gamma}^{\beta \mu} c_{\mu}^{\alpha\nu},
$$
and therefore the expression~\eqref{eq:tmp2} is equal to
\begin{align*}
&\frac{q_{\alpha}-q _{\beta} +\delta-2 q_{\mu}}{2}c_{\nu}^{\beta\mu}c_{\mu}^{\alpha \nu}+\frac{1}{2}e^\gamma\left(c_{\mu\gamma}^{\alpha \nu}c_{\nu}^{\beta \mu}-c_{\nu\gamma}^{\beta \mu} c_{\mu}^{\alpha\nu}\right)=\\
=&\frac{q_{\alpha}-q _{\beta} +\delta-2 q_{\mu}}{2}c_{\nu}^{\beta\mu}c_{\mu}^{\alpha \nu}+\frac{\delta-q_\alpha-q_\nu+q_\mu}{2}c_{\mu}^{\alpha \nu}c_{\nu}^{\beta \mu}-\frac{\delta-q_\beta-q_\mu+q_\nu}{2}c_{\nu}^{\beta \mu} c_{\mu}^{\alpha\nu}=\\
=&-\mu_\nu c_{\nu}^{\beta \mu}c_{\mu}^{\alpha\nu}=-\mu_\nu c_{\nu}^{\nu \mu}c_{\mu}^{\alpha\beta}. 
\end{align*}
It is sufficient to check that $\mu_{\nu}c_{\alpha\nu}^{\nu}=0$ for any $\alpha$. Indeed, we compute $X:=\mu_{\nu}c_{\alpha \nu}^{\nu} =\mu_\nu\eta ^{\nu\lambda} \eta _{\nu \theta} c^{\theta}_{\alpha \lambda}=- \mu_{\lambda } \eta ^{ \nu \lambda } \eta _{\nu \theta } c^{\theta} _{\alpha\lambda}= - \mu_{\lambda} c^{\lambda}_{ \alpha \lambda}=-X$, which implies that $X=0$, as required.

\medskip

Let us now prove equation~\eqref{eq:l=2}. Collecting together like terms we come to the equivalent equation
\begin{align*} 
&\frac{1}{8} \partial_{\gamma}\partial_{\nu}\left(g^{\mu \nu} c_{\mu}^{\alpha \beta}\right)+\frac{\mu_\alpha+\mu_\nu-1}{24}c_{\gamma}^{\alpha \nu} c_{\nu \mu}^{\beta \mu}=\frac{3-\mu_{\alpha}-2 \mu _{\beta}}{24} \partial_{\gamma}\left(c_{\nu}^{\nu \mu} c_{\mu}^{\alpha \beta}\right)+\frac{1}{24}\left[2 \partial_{\gamma}\left(c_{\nu \lambda}^{\alpha \lambda} g^{\nu \beta}\right)+g^{\alpha \nu} c_{\gamma \nu \lambda}^{\beta \lambda}\right],
\end{align*}
which is equivalent to
\begin{align*} 
&\frac{1-\frac{\delta}{2}-\mu_\nu}{8}\underbrace{\partial_{\gamma} \left(c_{\nu }^{\nu \mu} c_{\mu}^{\alpha \beta}\right)}_*+\frac{1}{8} \partial_{\gamma}\left(e^{\theta} \partial_{\nu}\left(c_{\theta}^{\mu \nu} c_{\mu}^{\alpha \beta}\right)\right)+\frac{\mu_\alpha+\mu_\nu-1}{24}c_{\gamma}^{\alpha \nu} c_{\nu \mu}^{\beta \mu}+\frac{\mu_{\alpha}+2 \mu _{\beta}-3}{24} \underbrace{\partial_{\gamma}\left(c_{\nu}^{\nu \mu} c_{\mu}^{\alpha \beta}\right)}_*\\
&-\frac{1}{24} \Big[ 2 \partial_{\gamma}\left(c_{\nu \lambda}^{\alpha \lambda} g^{\nu \beta}\right)+g^{\alpha \nu} c_{\gamma \nu \lambda}^{\beta \lambda}\Big]=0.	
\end{align*}
Using that $\mu_\nu c^{\nu\mu}_\nu=0$ we see that the left-hand side is equal to
\begin{align*} 
&\frac{1}{8}\partial_{\gamma}\left(e^{\theta} \partial_{\nu}\left(c_{\theta}^{\mu \nu} c_{\mu}^{\alpha \beta}\right)\right)+\frac{\mu_\alpha+\mu_\nu-1}{24}c_{\gamma}^{\alpha \nu} c_{\nu \mu}^{\beta \mu}+\frac{q_{\alpha}+2 q _{\beta}-3\delta}{24} \partial_{\gamma}\left(c_{\nu}^{\nu \mu} c_{\mu}^{\alpha \beta}\right)-\frac{1}{24}\left[2 \partial_{\gamma}\left({c}_{\nu \lambda}^{\alpha \lambda} g^{\nu \beta}\right)+g^{\alpha \nu} c_{\gamma \nu \lambda}^{\beta \lambda}\right].
\end{align*}
Transforming the first term in this expression as
\begin{align*}
&\partial_{\gamma}\left(e^{\theta} \partial_{\nu}\left(c_{\theta}^{\mu \nu} c_{\mu}^{\alpha \beta}\right)\right)=\partial_{\gamma}\left(e^{\theta} \partial_{\nu}\left(c_{\theta}^{\mu \beta} c_{\mu}^{\alpha \nu}\right)\right)=\partial_{\gamma}\left(e^{\theta} c_{\theta\nu}^{\mu \beta} c_{\mu}^{\alpha \nu}\right)+\partial_{\gamma}\left(e^{\theta}c_{\theta}^{\mu \beta} c_{\mu\nu}^{\alpha \nu}\right)=\\
=&(\delta-q_{\mu}-q_{\beta}+q_{\nu})\partial_{\gamma}\left(c_{\nu}^{\mu \beta} c_{\mu}^{\alpha \nu}\right)+\partial_{\gamma}\left(g^{\mu\beta} c_{\mu\nu}^{\alpha \nu}\right)=\\
=&(\delta-q_{\beta})\partial_{\gamma}\left(c_{\nu}^{\mu \beta} c_{\mu}^{\alpha \nu}\right)+\underbrace{(-q_{\mu}+q_{\nu})\partial_{\gamma}\left(c_{\nu}^{\mu \beta} c_{\mu}^{\alpha \nu}\right)}_{=0}+\partial_{\gamma}\left(g^{\mu\beta} c_{\mu\nu}^{\alpha \nu}\right)=\\
=&(\delta-q_{\beta})\partial_{\gamma}\left(c_{\nu}^{\mu \beta} c_{\mu}^{\alpha \nu}\right)+\partial_{\gamma}\left(g^{\mu\beta} c_{\mu\nu}^{\alpha \nu}\right),
\end{align*}
we come to the expression
\begin{align*} 
&\frac{\delta-q_{\beta}}{8}\underline{\partial_{\gamma}\left(c_{\nu}^{\mu\beta} c_{\mu}^{\alpha \nu}\right)}+\frac{1}{8}\boxed{\partial_{\gamma}\left(g^{\mu \beta} c_{\mu \nu}^{\alpha \nu}\right)}+\frac{\mu_\alpha+\mu_\nu-1}{24}c_{\gamma}^{\alpha \nu} c_{\nu \mu}^{\beta \mu}+\frac{q_{\alpha}+2 q _{\beta}-3\delta}{24} \underline{\partial_{\gamma}\left(c_{\nu}^{\nu \mu} c_{\mu}^{\alpha \beta}\right)}\\
&-\frac{1}{24}\left[2\boxed{\partial_{\gamma}\left({c}_{\nu \lambda}^{\alpha \lambda} g^{\nu \beta}\right)}+g^{\alpha \nu} c_{\gamma \nu \lambda}^{\beta \lambda}\right]=\\
=&\frac{\mu_\alpha-\mu_{\beta}}{24}\partial_{\gamma}\left(c_{\nu}^{\mu\beta} c_{\mu}^{\alpha \nu}\right)+\frac{\mu_\alpha+\mu_\nu-1}{24}c_{\gamma}^{\alpha \nu} c_{\nu \mu}^{\beta \mu}+\frac{1}{24}\left[\underline{\partial_{\gamma}\left({c}_{\nu \lambda}^{\alpha \lambda} g^{\nu \beta}\right)}-g^{\alpha \nu} c_{\gamma \nu \lambda}^{\beta \lambda}\right].
\end{align*}
Applying to the underlined term the formula
\begin{gather*} 
\d_\gamma g^{\nu\beta}=\left(1-q_{\gamma}\right)c_{\gamma}^{\nu \beta}+e^{\theta} c_{\theta \gamma}^{\nu \beta}=\left(1-\mu_{\nu}-\mu_{\beta}\right)c_{\gamma}^{\nu \beta},
\end{gather*}
we obtain
\begin{align*}
&\frac{\mu_\alpha-\mu_{\beta}}{24}\partial_{\gamma}\left(c_{\nu}^{\mu\beta} c_{\mu}^{\alpha \nu}\right)+\frac{\mu_\alpha+\mu_\nu-1}{24}c_{\gamma}^{\alpha \nu} c_{\nu \mu}^{\beta \mu}+\dfrac{1}{24}\left[c_{\nu \lambda \gamma}^{\alpha \lambda} g^{\nu \beta}+\left(1-\mu_{\nu}-\mu_{\beta}\right) c_{\nu \lambda}^{\alpha \lambda} c_{\gamma}^{\nu \beta}-g^{\alpha \nu} c^{ \beta \lambda }_{\gamma \nu \lambda}\right]=\\
=&\frac{\mu_\alpha-\mu_{\beta}}{24}\partial_{\gamma}\left(c_{\nu}^{\mu\beta} c_{\mu}^{\alpha \nu}\right)+\frac{\mu_\alpha+\mu_\nu-1}{24}c_{\gamma}^{\alpha \nu} c_{\nu \mu}^{\beta \mu}+\frac{1-\mu_{\nu}-\mu_{\beta}}{24}c_{\nu \lambda}^{\alpha \lambda} c_{\gamma}^{\nu \beta}+\dfrac{1}{24}\left[\underline{c_{\nu \lambda \gamma}^{\alpha \lambda} g^{\nu \beta}-g^{\alpha \nu} c^{ \beta \lambda }_{\gamma \nu \lambda}}\right].
\end{align*}
Expressing the underlined terms as follows:
\begin{align*}
&\d_\lambda\d_\gamma\underbrace{\left(c_\nu^{\alpha \lambda} g^{\nu \beta}-g^{\alpha \nu} c^{ \beta \lambda }_\nu \right)}_{=0}-c_{\nu \lambda }^{\alpha \lambda}\d_\gamma g^{\nu \beta}-c_{\nu\gamma}^{\alpha \lambda}\d_\lambda g^{\nu \beta}-c_{\nu}^{\alpha \lambda} \d_\lambda\d_\gamma g^{\nu \beta}+\d_\gamma g^{\alpha \nu} c^{ \beta \lambda }_{\nu \lambda}+\d_\lambda g^{\alpha \nu} c^{ \beta \lambda }_{\nu\gamma}+\d_\lambda\d_\gamma g^{\alpha \nu} c^{ \beta \lambda }_{\nu}=\\
=&\left(\mu_{\beta}+\mu_{\nu}-\mu_{\lambda}-\mu_{\alpha}\right)\d_\gamma\left(c_{\nu}^{\alpha \lambda}c_\lambda^{\beta \nu}\right)+\left(\mu_{\beta}+\mu_{\nu}-1\right)c^{\alpha \lambda}_{\nu \lambda}c_{\gamma}^{\beta \nu}+\left(1-\mu_{\alpha}-\mu_{\nu}\right)c_{\gamma}^{\alpha \nu} c_{\nu \lambda}^{\beta \lambda},
\end{align*}
we obtain
\begin{align*}
&\frac{\mu_\alpha-\mu_{\beta}}{24}\underbrace{\partial_{\gamma}\left(c_{\nu}^{\mu\beta} c_{\mu}^{\alpha \nu}\right)}_*+\frac{\mu_\alpha+\mu_\nu-1}{24}\underbrace{c_{\gamma}^{\alpha \nu} c_{\nu \mu}^{\beta \mu}}_{**}+\frac{1-\mu_{\nu}-\mu_{\beta}}{24}\underbrace{c_{\nu \lambda}^{\alpha \lambda} c_{\gamma}^{\nu \beta}}_{***}\\
&+\frac{\mu_{\beta}+\mu_{\nu}-\mu_{\lambda}-\mu_{\alpha}}{24}\underbrace{\d_\gamma\left(c_{\nu}^{\alpha \lambda}c_{\lambda}^{\beta \nu}\right)}_{*}+\frac{\mu_{\beta}+\mu_{\nu}-1}{24}\underbrace{c^{\alpha \lambda}_{\nu \lambda}c_{\gamma}^{\beta \nu}}_{***}+\frac{1-\mu_{\alpha}-\mu_{\nu}}{24}\underbrace{c_{\gamma}^{\alpha \nu} c_{\nu \lambda}^{\beta \lambda}}_{**}=\\
=&\frac{\mu_{\nu}-\mu_{\lambda}}{24}\d_\gamma\left(c_{\nu}^{\alpha \lambda}c_{\lambda}^{\beta \nu}\right)=0,
\end{align*}
as required.

\medskip

\noindent{\it\underline{Step 2}}. Let us prove~\eqref{eq:main equation} for $l=0$. Note that by the definition $\tu^\alpha(w^*_*,\eps)-w^\alpha\in\Im\,\d_x$. Therefore, Lemmas~\ref{lemma:constant terms of operators} and~\ref{lemma:DZconstant} imply that
$$
K_{2;\tu;0}^{\DZ;\alpha\beta}=\left(\frac{1}{2}-\mu_\beta\right)\eta^{\beta\nu}\sum_{m\ge 0}\frac{\d \tu^\alpha}{\d w^\rho_m}\eta^{\rho\theta}\d_x^{m+1}\Omega_{\theta,0;\nu,0}=\left(\frac{1}{2}-\mu_\beta\right)\eta^{\beta\nu}\sum_{m\ge 0}\left\{\tu^\alpha,\oh_{\nu,0}\right\}_{K_1^\DZ}.
$$
Lemmas~\ref{lemma:constant terms of operators} and~\ref{lemma:DRconstant} together with the fact that $\tu^\alpha(u^*_*,\eps)-u^\alpha\in\Im\,\d_x$ \cite[Lemma~7.1]{BDGR18} imply that
$$
K_{2;\tu;0}^{\DR;\alpha\beta}=\left(\frac{1}{2}-\mu_\beta\right)\eta^{\beta\nu}\sum_{m\ge 0}\frac{\d \tu^\alpha}{\d u^\rho_m}\eta^{\rho\theta}\d_x^{m+1}\frac{\delta\og_{\nu,0}}{\delta u^\theta}=\left(\frac{1}{2}-\mu_\beta\right)\eta^{\beta\nu}\sum_{m\ge 0}\left\{\tu^\alpha,\og_{\nu,0}\right\}_{K_1^\DR}.
$$
Since, in the coordinates~$\tu^\alpha$ and at the approximation up to~$\eps^2$, the local functionals $\oh_{\alpha,a}$ coincide with the local functionals $\og_{\alpha,a}$ and the Poisson operator~$K_1^\DZ$ coincide with the Poisson operator~$K_1^\DR$, we obtain $K_{2;\tu;0}^{\DZ;\alpha\beta}=K_{2;\tu;0}^{\DR;\alpha\beta}+O(\eps^4)$, as required.

\medskip

\noindent{\it\underline{Step 3}}. Let us finally prove that $K^{\DZ;[2]}_{2;\tu}=K^{\DR;[2]}_{2;\tu}$. Since we have proved~\eqref{eq:main equation} for $l=0,2,3$, the difference $K^{\DZ;[2]}_{2;\tu}-K^{\DR;[2]}_{2;\tu}$ has the form $R\d_x$, $R=(R^{\alpha\beta})$, where $R^{\alpha\beta}\in\cA_{\tu;2}$. Since the operators $K^{\DZ;[2]}_{2;\tu}$ and $K^{\DR;[2]}_{2;\tu}$ are skewsymmetric, we have
$$
(R\d_x)^\dagger=-R\d_x \quad \Leftrightarrow \quad R^T=R \text{ and } \d_x R=0.
$$
The property $\d_x R=0$ immediately implies that $R=0$, which completes the proof of the theorem.

\end{document}